\documentclass[journal]{IEEEtran}
\usepackage{amsmath,amsfonts,amssymb}
\usepackage{graphicx}
\usepackage{cite}
\usepackage{booktabs}
\usepackage{multirow}
\usepackage{array}
\usepackage{placeins}
\usepackage{algorithm}
\usepackage{algpseudocode}
\usepackage[colorlinks,citecolor=blue,linkcolor=blue]{hyperref}
\hypersetup{pdfauthor={Tuo Wu, Zhentian Zhang, Ye Tian, and Hing Cheung So}}
\graphicspath{{figures/}}

\newcommand{\R}{\mathbb R}
\newcommand{\E}{\mathbb E}
\newcommand{\tr}{\operatorname{tr}}
\newcommand{\diag}{\operatorname{diag}}
\newcommand{\rank}{\operatorname{rank}}
\newtheorem{proposition}{Proposition}
\newtheorem{lemma}{Lemma}
\newtheorem{corollary}{Corollary}
\newtheorem{remark}{Remark}
\newcommand{\LocEtBlockHalf}{0.39}
\newcommand{\LocStaticBlockHalf}{0.55}
\newcommand{\LocLosBlockHalf}{4.55}
\newcommand{\LocEtBlockFull}{0.49}
\newcommand{\LocStaticBlockFull}{0.65}
\newcommand{\LocOnlineJump}{0.65}
\newcommand{\LocOnlineEnd}{0.43}
\newcommand{\LocFrozenEnd}{0.63}
\newcommand{\LocStaticOnlineEnd}{0.99}
\newcommand{\LocMapJump}{2.03}
\newcommand{\LocMapEnd}{0.90}
\newcommand{\LocFrozenMapEnd}{1.91}

\newcommand{\LocWallStaticEnd}{1.01}
\newcommand{\LocWallEtEnd}{0.38}
\newcommand{\LocOneApEt}{0.65}
\newcommand{\LocThreeApEt}{0.27}
\newcommand{\LocPebOne}{0.45}
\newcommand{\LocPebFour}{0.27}
\newcommand{\LocLosAvailability}{27.60}
\newcommand{\LocLosConditional}{0.42}

\newcommand{\LocFairWindowPost}{0.45}
\newcommand{\LocFairRlsPost}{0.49}
\newcommand{\LocServiceLosHalf}{21.20}
\newcommand{\LocServiceStaticHalf}{58.20}
\newcommand{\LocServiceEtHalf}{84.80}
\newcommand{\LocServiceStaticFull}{40.00}
\newcommand{\LocServiceEtFull}{70.20}

\newcommand{\LocDiverseWall}{0.373}
\newcommand{\LocDiverseGp}{0.379}

\newcommand{\LocDiverseHuber}{0.380}

\newcommand{\LocKnnDense}{0.768}
\newcommand{\LocGpDense}{0.372}
\newcommand{\LocHuberDense}{0.384}
\newcommand{\LocContamLs}{1.932}
\newcommand{\LocContamHuber}{1.384}
\newcommand{\LocContamCauchy}{0.973}

\begin{document}

\title{Turning Multipath into Positioning Information: Persistent Electromagnetic-Twin DoA Localization}

\author{Tuo~Wu,
        Zhentian~Zhang,
        Ye~Tian,
        and~Hing~Cheung~So,~\IEEEmembership{Fellow,~IEEE}%
\thanks{\textit{Corresponding author: Ye Tian.}}
\thanks{T. Wu is with the School of Electronic and Information Engineering,
South China University of Technology, Guangzhou 510640, China
(e-mail: wutuo@scut.edu.cn).}
\thanks{Z. Zhang is with the National Mobile Communications Research
Laboratory, Frontiers Science Center for Mobile Information Communication
and Security, Southeast University, Nanjing 210096, China
(e-mail: zhangzhentian@seu.edu.cn).}
\thanks{Y. Tian is with the Faculty of Electrical Engineering and Computer
Science, Ningbo University, Ningbo 315211, China
(e-mail: tianye1@nbu.edu.cn).}
\thanks{H. C. So is with the Department of Electrical Engineering,
City University of Hong Kong, Hong Kong, China
(e-mail: hcso@ee.cityu.edu.hk).}}

\maketitle

\begin{abstract}
Reflected paths can support localization when line-of-sight (LoS) links are
blocked, but their usefulness depends on how accurately the stored model
predicts their directions. This paper develops electromagnetic-twin (ET)
direction-of-arrival (DoA) localization with a persistent, physically
calibrated path-angle field. A nominal floor plan supplies virtual-path
geometry, and sparse reference probes at known positions learn a compact
directional correction. Field order and regularization are selected from
calibration data alone. For each user, unordered DoAs are processed by joint
position estimation and one-to-one path association, using assignment-aware
multistart search and robust Gauss--Newton refinement. New trusted probes
update the stored correction through recursive least squares (RLS), without
reprocessing all previous measurements. The analysis establishes calibration
rank conditions, objective descent, and conditional information gains from
reflected paths, while also distinguishing the decay of initialization bias
from the convergence of localization error. Synthetic experiments use 500 queries per
operating point. At 50\% LoS blockage and $1.5^\circ$ DoA error, the
calibrated ET achieves \LocEtBlockHalf\ m root-mean-square error, versus
\LocStaticBlockHalf\ m for nominal geometry. After a wall change, its final
error is \LocOnlineEnd\ m, versus \LocFrozenEnd\ m for a frozen calibrated
map. Equal-evidence comparisons show competitive calibrated alternatives
and faster abrupt-change recovery with sliding-window updates than with
the configured RLS. The results support reusable directional calibration,
not a universally superior map or update rule, and validation is limited to
synthetic first-order reflections with trusted calibration labels.
\end{abstract}

\begin{IEEEkeywords}
Electromagnetic twin, direction-of-arrival, multipath-assisted localization,
virtual anchor, path association, robust positioning, Fisher information,
online calibration.
\end{IEEEkeywords}

\section{Introduction}
\IEEEPARstart{A}{ccurate} wireless localization supports location-aware
communications, indoor navigation, and integrated sensing and communication
services~\cite{wymeersch,liuISAC}. Antenna arrays provide direction-of-arrival
(DoA) measurements that relate a user's position to the known locations of
access points (APs). Combined angle and delay measurements can further support
position and orientation estimation in millimeter-wave systems
and three-dimensional indoor positioning~\cite{shahmansoori,wenJoint}.
The accuracy of these estimates depends not only on measurement noise but
also on the available propagation geometry, as established by information
analysis of localization accuracy~\cite{winFundamental}. In indoor environments, however,
walls and obstacles may block direct links while producing several reflected
arrivals. How to use these arrivals reliably is therefore an important
question for maintaining positioning service.

Multipath can be treated as either an impairment or a source of geometric
information. Rejecting non-line-of-sight (NLoS) measurements protects an
estimator designed for direct propagation, but may leave too few independent
bearings when line-of-sight (LoS) links are blocked. Alternatively, a
specular reflection can be represented through a mirrored anchor or user,
allowing the reflected direction to constrain position. Leitinger
\emph{et al.}~\cite{leitinger} quantified the position-related information in
multipath components, while Meissner \emph{et al.}~\cite{meissner} evaluated
multipath-assisted indoor navigation using measured ultra-wideband signals.
These results support using suitable reflected paths rather than discarding
all NLoS observations~\cite{witrisal}. For angle-based localization,
multipath-aware reference models likewise provide information beyond direct
bearings~\cite{wielandtAoA}.

When reflector locations are unknown, simultaneous localization and mapping
(SLAM) can estimate the user and environmental features together.
Channel-SLAM~\cite{gentner} and belief-propagation-based multipath
SLAM~\cite{geometricSLAM} exploit the consistency of propagation paths across
observations. Related work extends this principle to situational awareness
in millimeter-wave massive multiple-input multiple-output (MIMO)
systems~\cite{mendrzik}. These studies establish that propagation structure
can be learned and reused. A complementary situation arises when an initial
floor plan is already available but does not accurately predict the measured
directions. In that case, discovering the entire environment again may be
unnecessary, while the immediate task is to calibrate the useful paths in the
existing model.

Environment-dependent radio knowledge can also be retained without explicitly
estimating reflector geometry. Channel charting constructs a radio
representation from accumulated channel state information (CSI)
~\cite{studerc}, and channel knowledge maps (CKMs) organize location-indexed
channel descriptors for subsequent communication tasks~\cite{zengCKM}.
Digital-twin-assisted
fingerprinting further uses a digital replica to generate location-dependent
MIMO signatures~\cite{morais}. These approaches offer different interfaces
for reusing environmental information. Their accuracy in a physical site
nevertheless depends on the agreement between the retained representation
and the observations available during operation.

Digital-twin channel research addresses this agreement through synchronization
between physical observations and a digital propagation
model~\cite{wangDTC,liOnlineTwin}. Recent electromagnetic-twin (ET) studies
connect sparse evidence to an updated state and subsequent queries.
Wu \emph{et al.}~\cite{wuETCompletion} instantiate this interface through
channel-gain completion, while recurrent wireless-state updates support
AP association and beam selection~\cite{wuETPersistent}.
Persistent CSI estimation instead reuses angle--delay information across
pilot blocks~\cite{wuETCSI}, and ET beamforming converts sparse probes
into covariance-based sum-rate optimization~\cite{wuETBeamforming}.
These studies motivate retained, task-specific states but do not address
the path-angle calibration and unknown DoA association considered here.
We adopt this ET viewpoint for a directional field calibrated by reference
measurements and queried for user localization. A CKM can provide storage
and retrieval, while a broader twin may also represent scene objects,
devices, and network operation. Here, the maintained state contains only
the path-angle information needed for localization. The contribution is
its calibration and use with unordered DoAs, rather than map updating itself.

Consider an indoor deployment in which the AP locations are known and a
survey provides approximate wall positions. The nominal model predicts the
directions of direct and reflected arrivals throughout the room. A small
wall-location error can bias these predictions at many user locations,
even when the current DoAs are measured accurately. Known-position reference
probes reveal this discrepancy and can improve the model for users that were
not part of the calibration survey. If a reflecting surface subsequently
moves, the previously calibrated model can become stale. Reusing it
unchanged preserves that error, whereas rebuilding from only a few new
probes discards potentially useful earlier observations. This motivates a
directional state that is retained across queries and corrected as new
reference evidence becomes available.

Realizing this idea involves three coupled issues. First, a sparse probe
set cannot support an arbitrarily detailed correction. A low-order field may
miss spatial variation, whereas a richer field may fit calibration noise
rather than improve predictions. Second, the current user supplies an
unordered set of DoAs: the receiver must determine which modeled path
generated each observation while estimating the unknown position. Allowing
several observations to select the same path can give a misleading
initial estimate. Third, persistent operation must distinguish the
information used to locate a user from the information trusted to change
the model. Otherwise, an erroneous position or path label can be propagated
to subsequent queries.

To address these issues, we develop an ET-assisted multipath localization
framework with two separate operations: directional calibration and user
localization. The nominal floor plan supplies first-order reflection
geometry, and labelled probes at known positions fit a compact correction
to each predicted path angle. Calibration-only model selection chooses
between affine and quadratic fields. For a new user, a joint optimization
estimates position and one-to-one path assignments using the calibrated
field. Assignment-aware initialization and robust Gauss--Newton refinement
provide a practical solver for this mixed discrete--continuous problem.
After the queries at an update epoch, new trusted probes modify the stored
coefficients through recursive least squares (RLS). The corrected model
then serves the next users without requiring their positions to be known
or their trajectories to be tracked.

The study considers a two-dimensional region with known AP poses and
resolvable LoS and first-order reflected paths. Reference probes have known
positions and path labels, whereas both quantities are unknown for user
queries. This setting separates directional calibration from autonomous
reflector discovery and tests whether a compact maintained model can
improve localization under nominal-map mismatch, blockage, and environmental
change. The contributions are summarized as follows.
\begin{itemize}
\item We develop a persistent path-angle representation that combines
nominal reflection geometry, a calibrated spatial correction, and recursive
update statistics. Field order and regularization are selected from
reference measurements alone. Batch estimation and recursive updating
provide a common directional model for successive localization queries.
\item We formulate joint position estimation and one-to-one path association
as a constrained robust optimization problem. The proposed solver combines
assignment-aware multistart initialization, exact small-set association,
and backtracked Gauss--Newton refinement. We establish objective descent
and a retained-start safeguard under a common per-start budget.
\item We analyze correction-field identifiability, calibration error, and
conditional positioning information. The analysis explains how probe
geometry affects calibration and when independent reflected paths add
useful spatial constraints. It also distinguishes the decay of the initial
calibration bias from convergence of localization error.
\item We evaluate the framework using paired synthetic experiments with
500 localization queries per operating point. Comparisons with nominal
geometry, calibrated geometric and learned maps, and alternative robust
losses identify the benefit of directional calibration. Equal-evidence
online comparisons further distinguish maintaining an updated state from
the particular memory policy used to update it.
\end{itemize}

The remainder of this paper is organized as follows.
Section~\ref{sec:system} presents the system model.
Section~\ref{sec:calibration} develops directional calibration.
Sections~\ref{sec:problem} and~\ref{sec:algorithm} formulate and solve the
localization problem. Section~\ref{sec:information} analyzes positioning
information. Section~\ref{sec:numerical} presents numerical results, and
Section~\ref{sec:conclusion} concludes the paper.

\section{System Model}
\label{sec:system}
In this section, we describe the directional measurements and the
environment-dependent state used for localization. We then specify how
reference probes update this state without using unknown user positions.

\subsection{Propagation Geometry and DoA Measurements}
As illustrated in Fig.~\ref{fig:geometry}, we consider a two-dimensional
indoor region $\Omega\subset\R^2$ served by $A$ APs at known positions
$\mathbf q_a$, $a=1,\ldots,A$. The AP array orientations and nominal
floor plan $\mathcal M_0$ are known in a common coordinate frame.
A user at unknown position $\mathbf p_t=[p_{x,t},p_{y,t}]^T$ is observed
through LoS propagation and a finite set of first-order wall reflections.
Let $\ell=0$ denote the LoS path and $\ell=1,\ldots,L_a-1$ the reflected-path
candidates at AP $a$.

The array front end returns the unordered DoA set
\begin{equation}
\mathcal Z_{a,t}=\{(z_{am,t},\sigma_{am,t}^2)\}_{m=1}^{M_{a,t}},
\label{eq:doa-set}
\end{equation}
where $z_{am,t}$ and $\sigma_{am,t}^2$ denote the measured angle and its
noise variance, respectively. For a measurement generated by path $\ell$,
\begin{equation}
z_{am,t}=\theta_{a\ell}(\mathbf p_t;\mathcal M^\star)
+n_{am,t},\qquad n_{am,t}\sim\mathcal N(0,\sigma_{am,t}^2),
\label{eq:measurement}
\end{equation}
where $\mathcal M^\star$ is the physical geometry, which need not coincide
with $\mathcal M_0$. Angles are wrapped consistently to $(-\pi,\pi]$.
The numerical implementation uses the equivalent degree convention.
The path labels are unknown at query time, and LoS may be absent.
We assume resolvable modeled paths and no false detections, so that
$M_{a,t}\leq L_a$. This angle-domain model isolates map calibration and
positioning from array processing, and it does not establish that all paths
are resolvable by a given array. The experiments use a common angular
variance, while heterogeneous variances are allowed in the formulation.
The epoch index is omitted below when a single query is considered.

\begin{figure}[t]
\centering
\includegraphics[width=.96\columnwidth]{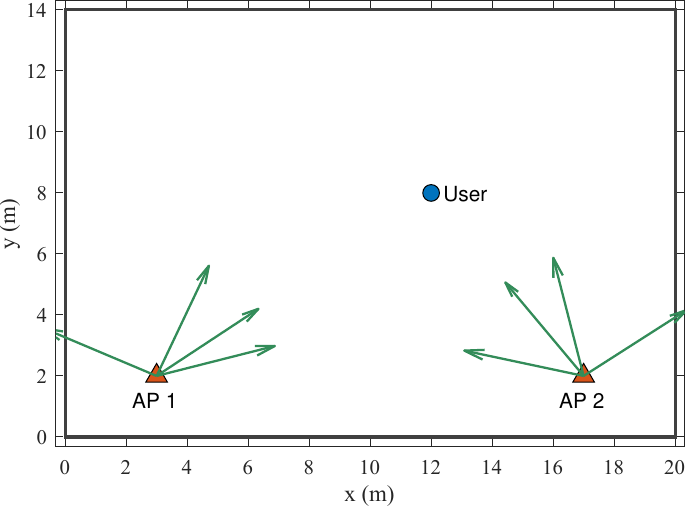}
\caption{Two-AP indoor model with direct and reflected arrival directions.}
\label{fig:geometry}
\end{figure}

\subsection{Calibrated Directional Field}
For a planar reflector, the virtual-user position is
\begin{equation}
\mathbf v_{a\ell}(\mathbf p;\mathcal M)
=\mathbf H_{a\ell}(\mathcal M)\mathbf p+\mathbf h_{a\ell}(\mathcal M).
\label{eq:virtual-user}
\end{equation}
The LoS path corresponds to $\mathbf H_{a0}=\mathbf I$ and
$\mathbf h_{a0}=\mathbf0$.
Reflection in a vertical wall $x=x_w$ gives
$\mathbf v_{a\ell}=[2x_w-p_x,p_y]^T$, and a horizontal wall is treated
analogously. The corresponding arrival direction is
\begin{equation}
\theta_{a\ell}^{\rm geo}(\mathbf p;\mathcal M)
=\operatorname{atan2}\!\left(
[\mathbf v_{a\ell}-\mathbf q_a]_2,
[\mathbf v_{a\ell}-\mathbf q_a]_1\right).
\label{eq:geometric-angle}
\end{equation}
Instead of estimating every physical wall parameter, we calibrate the
direction predicted from the nominal map:
\begin{equation}
\theta_{a\ell}^{\rm ET}(\mathbf p)
=\theta_{a\ell}^{\rm geo}(\mathbf p;\mathcal M_0)
+\boldsymbol\phi^T(\mathbf p)\boldsymbol\beta_{a\ell}.
\label{eq:et-angle}
\end{equation}
Here, $\boldsymbol\beta_{a\ell}$ contains the correction coefficients.
For nominal room width $W$ and height $H$, define
$u=(p_x-W/2)/W$ and $v=(p_y-H/2)/H$. We consider
\begin{equation}
\begin{split}
\boldsymbol\phi_1(\mathbf p)&=[1,u,v]^T,\\
\boldsymbol\phi_2(\mathbf p)&=[1,u,v,u^2,uv,v^2]^T.
\end{split}
\label{eq:features}
\end{equation}
The affine field represents angular offset and spatial slope, while the
quadratic field also captures curvature. Both preserve the nominal
reflection geometry and remain linear in their coefficients. The order
$r\in\{1,2\}$, and hence dimension $D\in\{3,6\}$, is selected from
calibration data in Section~\ref{sec:calibration}. We subsequently write
$\boldsymbol\phi=\boldsymbol\phi_r$.

\subsection{Persistent State and Update Protocol}
\label{sec:trusted-update}
At the beginning of epoch $t$, the implemented ET stores
\begin{equation}
\mathcal T_t=\{\mathcal M_0,r,\lambda_\beta,\mathbf B_t,
\{\mathbf A_{a\ell,t}\}_{a,\ell}\},
\label{eq:twin-state}
\end{equation}
where $\mathbf B_t=\{\boldsymbol\beta_{a\ell,t}\}$,
$\lambda_\beta$ is the initial ridge strength, and
$\mathbf A_{a\ell,t}$ is the inverse weighted normal matrix for recursive
calibration, not a coefficient-error covariance.
Users are first localized with the current field. New reference probes
with known positions and path labels then update the coefficients and
normal matrices for the next epoch. The same state can therefore serve
different users without tracking a particular trajectory.

Only trusted reference probes update the reported ET. Estimated user
positions are not reused as calibration labels. Self-calibration would
require independently validated position confidence, association
consistency, and innovation screening. It is not part of the tested
method. Motion tracking, delay measurements~\cite{wenJoint},
and controlled reconfigurable intelligent surface (RIS) paths are likewise extensions,
not additional sources of information in the reported results.
Setting $\mathbf B_t=\mathbf0$ recovers the nominal map, while freezing
$\mathbf B_t$ after initial calibration gives the frozen ET. The
physical-geometry reference employs $\mathcal M^\star$ and serves only as a
benchmark, not a guaranteed finite-sample error bound.

\section{Sparse-Probe Directional Calibration}
\label{sec:calibration}
We first estimate the correction field from a reference batch, and then
derive its recursive update. All model-selection decisions use reference
measurements rather than localization errors at test users.

\subsection{Batch Estimation and Model Selection}
Suppose $K$ reference probes provide known positions $\mathbf p_k$ and
path-labeled angles $\widetilde\theta_{a\ell,k}$. Their angular innovations
relative to the nominal geometry are
\begin{equation}
d_{a\ell,k}=\operatorname{wrap}\!\left[
\widetilde\theta_{a\ell,k}
-\theta_{a\ell}^{\rm geo}(\mathbf p_k;\mathcal M_0)\right].
\label{eq:calibration-innovation}
\end{equation}
Stacking $\boldsymbol\phi^T(\mathbf p_k)$ as the rows of
$\boldsymbol\Phi_K$ gives the regression model
\begin{equation}
\mathbf d_{a\ell}=\boldsymbol\Phi_K\boldsymbol\beta_{a\ell}
+\boldsymbol\epsilon_{a\ell}.
\label{eq:calibration-linear}
\end{equation}
This is a local angular model. Innovations crossing a wrapping
discontinuity require unwrapping or circular regression. Reference path
labels are supplied by the simulator, whereas query labels are unknown.
Obtaining these reference labels from real measurements remains a
separate calibration requirement.

For fixed basis and regularization $\lambda_\beta>0$, we solve
\begin{equation}
\widehat{\boldsymbol\beta}_{a\ell}
=\arg\min_{\boldsymbol\beta}
\|\mathbf d_{a\ell}-\boldsymbol\Phi_K\boldsymbol\beta\|_2^2
+\lambda_\beta\|\boldsymbol\beta\|_2^2.
\label{eq:ridge-problem}
\end{equation}
The objective is strictly convex. Setting its gradient to zero yields
\begin{equation}
\widehat{\boldsymbol\beta}_{a\ell}
=(\boldsymbol\Phi_K^T\boldsymbol\Phi_K+\lambda_\beta\mathbf I)^{-1}
\boldsymbol\Phi_K^T\mathbf d_{a\ell}.
\label{eq:ridge-solution}
\end{equation}
Regularization stabilizes a sparse fit, while its strength and the basis
order determine the bias--variance tradeoff. We select them jointly by
leave-one-out cross-validation (LOOCV)~\cite{hastieESL}.
For candidate $(r,\lambda)$, define the hat matrix
\begin{equation}
\mathbf H_{r,\lambda}
=\boldsymbol\Phi_r(\boldsymbol\Phi_r^T\boldsymbol\Phi_r+
\lambda\mathbf I)^{-1}\boldsymbol\Phi_r^T
\label{eq:hat-matrix}
\end{equation}
and residual $\mathbf e_{a\ell}=(\mathbf I-\mathbf H_{r,\lambda})
\mathbf d_{a\ell}$. Removing probe $k$ subtracts
$\boldsymbol\phi_k\boldsymbol\phi_k^T$ from the normal matrix.
The Sherman--Morrison identity therefore gives the held-out residual
\begin{equation}
e_{a\ell,k}^{(-k)}
=\frac{e_{a\ell,k}}{1-[\mathbf H_{r,\lambda}]_{kk}}.
\label{eq:press-residual}
\end{equation}
Hence one fit per candidate suffices to evaluate
\begin{equation}
{\rm CV}(r,\lambda)
=\frac{1}{K\sum_aL_a}\sum_{a,\ell,k}
\left(e_{a\ell,k}^{(-k)}\right)^2.
\label{eq:press-score}
\end{equation}
We minimize \eqref{eq:press-score} over
$\lambda\in\{10^{-5},\ldots,10^{-1},1\}$
and eligible orders, using one selected pair
$(\widehat r,\widehat\lambda)$ for all APs and paths.
The affine basis is always eligible, while the quadratic basis requires
$K\geq8$, leaving at least seven probes for six coefficients in each
held-out fit. This is an eligibility rule, not a coverage guarantee.
The selected model is refitted to the full initial batch and then kept
at the same order during online operation. With no probes, its correction
is zero. LOOCV predicts held-out angles, not unknown user positions.
It therefore provides no guarantee of improved localization
root-mean-square error (RMSE).

\subsection{Recursive Update}
For a fixed selected basis, subsequent reference probes can be assimilated
without repeating the batch fit. Suppressing path indices and indexing
individual probes by $k$, RLS~\cite{haykin} with forgetting factor
$0<\lambda_f\leq1$ gives
\begin{align}
\mathbf k_k&=\frac{\mathbf A_{k-1}\boldsymbol\phi_k}
{\lambda_f+\boldsymbol\phi_k^T\mathbf A_{k-1}\boldsymbol\phi_k},\\
\widehat{\boldsymbol\beta}_k
&=\widehat{\boldsymbol\beta}_{k-1}
+\mathbf k_k(d_k-\boldsymbol\phi_k^T\widehat{\boldsymbol\beta}_{k-1}),\\
\mathbf A_k&=\lambda_f^{-1}
(\mathbf I-\mathbf k_k\boldsymbol\phi_k^T)\mathbf A_{k-1}.
\label{eq:rls}
\end{align}
The initial inverse normal matrix is
$\mathbf A_0=(\boldsymbol\Phi_K^T\boldsymbol\Phi_K+
\widehat\lambda\mathbf I)^{-1}$, consistent with
\eqref{eq:ridge-solution}. Thus RLS continues the same fit.
For $\lambda_f<1$, previous normal-equation terms, including the initial
ridge contribution, are progressively discounted.

\begin{proposition}[Decay of initialization bias]
\label{prop:forgetting}
Suppose the post-change correction is exactly representable by a constant
$\boldsymbol\beta^\star$, the features are bounded and deterministic, and
subsequent calibration noise is conditionally zero mean. For
$0<\lambda_f<1$, assume persistent excitation:
\begin{equation}
\sum_{i=k}^{k+N_{\rm pe}-1}\boldsymbol\phi_i\boldsymbol\phi_i^T
\succeq\mu\mathbf I,\qquad \mu>0.
\label{eq:persistent-excitation}
\end{equation}
Reset the probe index at the change and let
$\mathbf G_k=\mathbf A_k^{-1}$. Conditional on the selected basis,
initial state, and features,
\begin{equation}
\E[\widehat{\boldsymbol\beta}_k-\boldsymbol\beta^\star]
=\lambda_f^k\mathbf G_k^{-1}\mathbf G_0
(\widehat{\boldsymbol\beta}_0-\boldsymbol\beta^\star).
\label{eq:rls-mean-error}
\end{equation}
For $k\geq N_{\rm pe}$, its norm is at most
$\lambda_f^{k-N_{\rm pe}+1}
\|\mathbf G_0(\widehat{\boldsymbol\beta}_0-\boldsymbol\beta^\star)\|/\mu$.
\end{proposition}
\begin{IEEEproof}
The weighted normal equation retains the initial state with weight
$\lambda_f^k$, and taking expectation yields \eqref{eq:rls-mean-error}.
The most recent $N_{\rm pe}$ probes imply
$\mathbf G_k\succeq\lambda_f^{N_{\rm pe}-1}\mu\mathbf I$.
Applying this bound to $\mathbf G_k^{-1}$ proves the claim.
Appendix~\ref{app:calibration} gives the full recursion.
\end{IEEEproof}
Proposition~\ref{prop:forgetting} concerns decay of initialization bias,
not monotonic sample-path RMSE or zero steady-state error. Measurement
noise and finite-basis approximation error remain.

\begin{remark}[Trusted updates]
\label{rem:trusted}
The implemented recursion accepts only trusted probes and uses no change
detector or innovation gate. With uncertain labels, isolated and persistent residuals
may indicate outliers and a stale map, respectively,
but systematic sensor error can produce either pattern. An uncertainty-based gate would require separate validation and must not be treated as
an established benefit of the present ET.
\end{remark}

\subsection{Calibration Accuracy and Identifiability}
For fixed basis and ridge, exact reference positions, correct labels,
and homoscedastic noise variance $\sigma_c^2$, the sampling covariance is
\begin{equation}
\begin{split}
\mathbf P_{\beta,a\ell}
={}&\sigma_c^2(\boldsymbol\Phi_K^T\boldsymbol\Phi_K+
\lambda_\beta\mathbf I)^{-1}\boldsymbol\Phi_K^T\boldsymbol\Phi_K\\
&\mathord{}\cdot(\boldsymbol\Phi_K^T\boldsymbol\Phi_K+
\lambda_\beta\mathbf I)^{-1}.
\end{split}
\label{eq:beta-covariance}
\end{equation}
This follows by substituting \eqref{eq:calibration-linear} into
\eqref{eq:ridge-solution}. The corresponding ridge bias is
$-\lambda_\beta(\boldsymbol\Phi_K^T\boldsymbol\Phi_K+
\lambda_\beta\mathbf I)^{-1}\boldsymbol\beta_{a\ell}$.
Thus covariance alone is not mean-square error (MSE). Furthermore,
\eqref{eq:beta-covariance} holds for a prespecified model. It does not
account for choosing the basis and ridge from the same noisy batch.
It is used for conditional analysis, not for weighting simulated queries.

\begin{lemma}[Identifiable correction field]
\label{lem:calibration-rank}
In the noise-free unregularized model \eqref{eq:calibration-linear},
$\boldsymbol\beta_{a\ell}\in\R^D$ is uniquely identifiable if and only if
\begin{equation}
\rank(\boldsymbol\Phi_K)=D.
\label{eq:calibration-rank}
\end{equation}
\end{lemma}
\begin{IEEEproof}
Two coefficient vectors produce identical innovations precisely when
their difference belongs to the null space of $\boldsymbol\Phi_K$.
Uniqueness is therefore equivalent to a trivial null space.
\end{IEEEproof}
At least $D$ probes are necessary, but spatial coverage is also needed.
For example, collinear probes leave directions unobserved in either
basis. Ridge ensures a unique numerical fit without restoring this
missing information.

Let $\mu_K=\lambda_{\min}(\boldsymbol\Phi_K^T\boldsymbol\Phi_K)>0$.
For a fixed full-rank basis, an exactly specified correction, and no ridge,
\begin{equation}
\begin{split}
\E\|\widehat{\boldsymbol\beta}_{a\ell}-\boldsymbol\beta_{a\ell}\|_2^2
&=\sigma_c^2\tr[(\boldsymbol\Phi_K^T\boldsymbol\Phi_K)^{-1}]\\
&\leq D\sigma_c^2/\mu_K.
\end{split}
\label{eq:calibration-bound}
\end{equation}
The derivation is provided in Appendix~\ref{app:calibration}.
This identity explains the role of probe geometry, rather than the MSE
of the selected ridge estimator. Once the feature directions are well
conditioned, additional probes mainly average noise. They do not remove
finite-basis bias. This distinction motivates the probe-budget and
wall-mismatch experiments.

\section{Joint Localization and Path Association}
\label{sec:problem}
Given a calibrated field, our objective is to estimate a user position
and identify the paths that generated its unordered DoAs. The field is
held fixed during each query.

\subsection{Association Constraints}
Let $x_{am\ell}=1$ if measurement $m$ at AP $a$ is assigned to path
$\ell$, and zero otherwise. Each measurement has one label, whereas a
path explains at most one measurement:
\begin{equation}
\sum_{\ell=0}^{L_a-1}x_{am\ell}=1,\qquad
\sum_{m=1}^{M_a}x_{am\ell}\leq1,\qquad x_{am\ell}\in\{0,1\}.
\label{eq:association-constraints}
\end{equation}
The inequality permits missing paths, including blocked LoS propagation.
Denote the resulting assignment set by $\mathcal X$. Unlike probabilistic
data association, which retains weighted origin hypotheses~\cite{barshalom},
we optimize explicit one-to-one labels.
Unmodeled detections would require dummy assignments with an outlier
penalty. Such detections are excluded from the present experiment.

\subsection{Problem Formulation}
The angular residual for candidate association $(a,m,\ell)$ is
\begin{equation}
r_{am\ell}(\mathbf p)=\operatorname{wrap}
[z_{am}-\theta_{a\ell}^{\rm ET}(\mathbf p)].
\label{eq:residual}
\end{equation}
We use the Huber loss
\begin{equation}
\rho_\delta(r)=
\begin{cases}
r^2/2,&|r|\leq\delta,\\
\delta|r|-\delta^2/2,&|r|>\delta,
\end{cases}
\label{eq:huber-loss}
\end{equation}
which is quadratic for small errors and grows linearly for large ones.
For positive scales $\nu_{am\ell}^2$ fixed throughout a query, the joint
problem is
\begin{subequations}
\begin{align}
\text{(P1)}\quad\min_{\mathbf p,\mathbf x}\quad&
\sum_{a,m,\ell}x_{am\ell}
\frac{\rho_\delta(r_{am\ell}(\mathbf p))}{\nu_{am\ell}^2}
\label{eq:joint-objective}\\
\mathrm{s.t.}\quad&\mathbf p\in\Omega,\quad\mathbf x\in\mathcal X.
\end{align}
\label{eq:joint-problem}
\end{subequations}
Denote its objective by $f(\mathbf p,\mathbf x)$.
The experiments use one common scale, which cancels from the minimizer.
A weighted extension could evaluate the variance in
\eqref{eq:effective-variance} at a coarse initializer and freeze the
weights, but is not tested here.

Problem~\eqref{eq:joint-problem} is nonconvex due to both nonlinear
angle geometry and discrete association. Nevertheless, fixing position
gives a linear assignment problem, and fixing association leaves a
two-dimensional continuous problem. We exploit this structure below. The convexity of the Huber loss in its scalar argument does not make
the joint problem convex.

\section{Proposed Localization Algorithm}
\label{sec:algorithm}
We initialize several candidate positions and alternate exact path
assignment with robust position refinement. All candidates minimize the
same objective in \eqref{eq:joint-objective}. Algorithm~\ref{alg:localization}
summarizes the procedure.

\subsection{Assignment-Aware Initialization}
On a grid $\mathcal G\subset\Omega$, precompute
$\theta_{a\ell}^{\rm ET}(\mathbf g)$ for all AP--path pairs. A
low-cost nearest-path initializer is
\begin{equation}
\mathbf p_{\rm near}^{(0)}=\arg\min_{\mathbf g\in\mathcal G}
\sum_{a,m}\min_\ell
\operatorname{wrap}^2[z_{am}-\theta_{a\ell}^{\rm ET}(\mathbf g)].
\label{eq:grid-init}
\end{equation}
This score permits several measurements to select the same path and
differs from the robust query objective. We retain its solution as one
start, but rank additional starts using
\begin{equation}
F(\mathbf g)=\min_{\mathbf x\in\mathcal X}f(\mathbf g,\mathbf x).
\label{eq:assignment-grid}
\end{equation}
Starting with $\mathbf p_{\rm near}^{(0)}$, we scan grid points in
ascending order of $F$ and retain a point only if it is at least 2 m
from every existing start, until $S\leq3$ starts are obtained.
This separation diversifies the local search rather than detecting
environmental change. Each start receives the same refinement budget.
For their final positions $\widehat{\mathbf p}_s$, select
\begin{equation}
\widehat{\mathbf p}
=\arg\min_{\mathbf p\in\{\widehat{\mathbf p}_1,\ldots,
\widehat{\mathbf p}_S\}}F(\mathbf p).
\label{eq:multistart-choice}
\end{equation}
The precomputed angles remain valid until the field is updated, whereas
assignment costs are query dependent. Nominal, calibrated-wall, Gaussian-process (GP),
and physical-geometry maps use the same candidate rule and budget.

\subsection{Alternating Assignment and Position Updates}
\subsubsection{Path assignment}
For fixed $\mathbf p^{(i)}$, each AP independently minimizes its sum of
costs $\rho_\delta(r_{am\ell}(\mathbf p^{(i)}))/\nu_{am\ell}^2$
subject to \eqref{eq:association-constraints}. This rectangular assignment
can be solved by the Hungarian algorithm~\cite{kuhn}. With four candidate
paths, our implementation instead enumerates at most $4!=24$ injections.
Let $\widehat\ell_{am}^{(i)}$ denote the resulting labels.

\subsubsection{Position refinement}
For fixed labels, stack the residuals into $\mathbf r(\mathbf p)$ and
define
\begin{equation}
\mathbf J(\mathbf p)=
\frac{\partial\boldsymbol\theta_{\rm assoc}^{\rm ET}(\mathbf p)}
{\partial\mathbf p^T}.
\label{eq:jacobian}
\end{equation}
For a displacement vector $\mathbf s\in\R^2$,
$\mathbf r(\mathbf p+\mathbf s)\simeq\mathbf r(\mathbf p)-\mathbf J\mathbf s$
locally. Iteratively reweighted least squares~\cite{huber,nocedal} adopts
\begin{equation}
w_j^{(i)}=\frac{1}{\nu_j^2}
\min\left\{1,\frac{\delta}{|r_j^{(i)}|}\right\},
\label{eq:irls-weights}
\end{equation}
with its continuous limit at $r_j^{(i)}=0$. Writing
$\mathbf W^{(i)}=\diag(w_j^{(i)})$, the damped linearized subproblem is
\begin{equation}
\min_{\mathbf s}\ \frac12
\|(\mathbf W^{(i)})^{1/2}(\mathbf r-\mathbf J\mathbf s)\|_2^2
+\frac{\epsilon}{2}\|\mathbf s\|_2^2,\qquad\epsilon>0.
\label{eq:gn-subproblem}
\end{equation}
The minimizing displacement, denoted by $\Delta\mathbf p^{(i)}$, satisfies
\begin{equation}
\Delta\mathbf p^{(i)}
=[\mathbf J^T\mathbf W^{(i)}\mathbf J+\epsilon\mathbf I]^{-1}
\mathbf J^T\mathbf W^{(i)}\mathbf r.
\label{eq:gn-step}
\end{equation}
The geometric Jacobian is derived in Appendix~\ref{app:jacobian}.
Because \eqref{eq:gn-subproblem} is a local approximation, we test
the original robust cost by backtracking:
\begin{equation}
\mathbf p^{(i+1)}
=\Pi_\Omega[\mathbf p^{(i)}+\alpha_i\Delta\mathbf p^{(i)}],
\qquad 0<\alpha_i\leq1.
\label{eq:position-update}
\end{equation}
A trial is accepted only if the fixed-assignment objective does not
increase. The implementation caps the direction norm at 1.5 m and tests
up to eight successively halved steps. If none is accepted, position is
unchanged. Each start stops after ten iterations or displacement below
$10^{-4}$ m. Assignment, refinement, and acceptance use the same loss. The linearized solve alone is not a global descent guarantee.

\begin{algorithm}[t]
\caption{ET-Assisted Multipath Localization}
\label{alg:localization}
\begin{algorithmic}[1]
\Require DoAs, fixed calibrated field, grid, and query scales.
\Ensure Position $\widehat{\mathbf p}$ and path labels.
\State Retain \eqref{eq:grid-init} and add separated starts using
\eqref{eq:assignment-grid}.
\For{each retained start}
\Repeat
\State Minimize assignment cost for fixed position.
\State Form weights \eqref{eq:irls-weights} and direction
\eqref{eq:gn-step}.
\State Project and backtrack using \eqref{eq:position-update}.
\Until{the stopping tolerance or iteration limit is met}
\State Store the refined position and its cost $F$.
\EndFor
\State Select \eqref{eq:multistart-choice} and return its minimizing labels.
\end{algorithmic}
\end{algorithm}

\subsection{Convergence and Complexity}
\begin{proposition}[Objective descent]
\label{prop:descent}
If assignment is solved exactly and each position update is accepted only
when its fixed-assignment cost does not increase, the sequence of
objective values is nonincreasing and convergent.
\end{proposition}
\begin{IEEEproof}
Exact reassignment and accepted backtracking respectively imply
\begin{equation}
\begin{split}
f(\mathbf p^{(i+1)},\mathbf x^{(i+1)})
&\leq f(\mathbf p^{(i+1)},\mathbf x^{(i)})\\
&\leq f(\mathbf p^{(i)},\mathbf x^{(i)}).
\end{split}
\label{eq:descent-chain}
\end{equation}
The objective is nonnegative, so the scalar sequence converges.
\end{IEEEproof}
This result does not establish convergence to a globally optimal
position or to a stationary point under all stopping rules.

\begin{corollary}[Retained-start safeguard]
\label{cor:retained-start}
For a fixed map, loss, and per-start budget, let
$\widehat{\mathbf p}_{\rm near}$ be the refined result from
\eqref{eq:grid-init} alone. The proposed selection satisfies
\begin{equation}
F(\widehat{\mathbf p})\leq F(\widehat{\mathbf p}_{\rm near}).
\label{eq:retained-safeguard}
\end{equation}
\end{corollary}
\begin{IEEEproof}
That same refinement is included among the candidates in
\eqref{eq:multistart-choice}.
\end{IEEEproof}
A smaller fitting cost need not imply smaller position error under
noisy or mismatched angles. The rule is fixed before evaluating queries.

For $G$ grid points, $A$ APs, and at most $L$ paths and $M\leq L$
observations, nearest-path scanning costs $\mathcal O(GALM)$.
Exact assignment adds $\mathcal O(GAL^3)$ with Hungarian matching, or
$\mathcal O(GAM L!/(L-M)!)$ with the enumeration used here.
Sorting costs $\mathcal O(G\log G)$ and candidate screening costs $\mathcal O(GS)$.
Each of the $S\leq3$ local solves uses a $2\times2$ linear system plus
assignment and line-search evaluations per iteration. The improved
initialization therefore has an explicit computational cost.
The field stores $D\in\{3,6\}$ coefficients per path, and each RLS
path update costs $\mathcal O(D^2)$. These dimensions support the small-path
implementation but do not establish scalability of enumeration.

\section{Positioning Information Analysis}
\label{sec:information}
We next relate calibration quality and reflected-path geometry to local
positioning information. The analysis assumes correct association and a
specified error model. It is distinct from the robust estimator's
finite-sample performance.

\subsection{Angular Uncertainty and Fisher Information}
For the fixed-model coefficient covariance in
\eqref{eq:beta-covariance}, the induced angle variance is
\begin{equation}
\sigma_{{\rm map},a\ell}^2(\mathbf p)
=\boldsymbol\phi^T(\mathbf p)\mathbf P_{\beta,a\ell}
\boldsymbol\phi(\mathbf p).
\label{eq:map-variance}
\end{equation}
If front-end, map, and residual model errors are independent, their
variances add:
\begin{equation}
\sigma_{{\rm eff},am\ell}^2(\mathbf p)
=\sigma_{am}^2+\sigma_{{\rm map},a\ell}^2(\mathbf p)
+\sigma_{{\rm model},a\ell}^2.
\label{eq:effective-variance}
\end{equation}
Here, the last term represents random residual error, not deterministic
wall bias. Appendix~\ref{app:uncertainty} separates variance and squared bias.
This uncertainty model is not used to weight the simulated queries.

For an available correctly associated AP--path pair $(a,\ell)\in\mathcal A$,
define
\begin{equation}
\mathbf g_{a\ell}(\mathbf p)
=\nabla_{\mathbf p}\theta_{a\ell}^{\rm ET}(\mathbf p).
\label{eq:angle-gradient}
\end{equation}
Assume independent locally Gaussian errors away from wrapping
discontinuities, with mean $\theta_{a\ell}^{\rm ET}(\mathbf p)$ and
variance $v_{a\ell}(\mathbf p)=\sigma_{{\rm eff},a\ell}^2(\mathbf p)$.
Freezing the variances locally gives the mean-gradient information
\begin{equation}
\mathbf F(\mathbf p)=\sum_{(a,\ell)\in\mathcal A}
\frac{\mathbf g_{a\ell}(\mathbf p)\mathbf g_{a\ell}^T(\mathbf p)}
{\sigma_{{\rm eff},a\ell}^2(\mathbf p)}.
\label{eq:fim}
\end{equation}
This is the full Gaussian Fisher information matrix (FIM) for
position-independent variances, as in the ideal-geometry experiment.
Otherwise, the full Gaussian information is~\cite{kay}
\begin{equation}
\mathbf F_{\rm G}(\mathbf p)=\mathbf F(\mathbf p)
+\frac12\sum_{(a,\ell)\in\mathcal A}
\frac{\boldsymbol\eta_{a\ell}(\mathbf p)\boldsymbol\eta_{a\ell}^T(\mathbf p)}
{v_{a\ell}^2(\mathbf p)}.
\label{eq:full-gaussian-fim}
\end{equation}
Here, $\boldsymbol\eta_{a\ell}(\mathbf p)
=\nabla_{\mathbf p}v_{a\ell}(\mathbf p)$ is the variance-gradient vector.
For nonsingular $\mathbf F$ and constant variances, the ideal-model
position error bound (PEB) is
\begin{equation}
\operatorname{PEB}(\mathbf p)=\sqrt{\tr(\mathbf F^{-1}(\mathbf p))}.
\label{eq:peb}
\end{equation}
As the Cram\'er--Rao bound (CRB), it applies to locally unbiased estimation
under the assumed likelihood.

\subsection{Reflected-Path Information and Local Identifiability}
\begin{proposition}[Information from a reflected path]
\label{prop:path-information}
With existing paths and variances unchanged, adding an independent,
correctly associated path changes the fixed-variance FIM by
\begin{equation}
\Delta\mathbf F
=\frac{\mathbf g_{a\ell}\mathbf g_{a\ell}^T}
{\sigma_{{\rm eff},a\ell}^2}\succeq\mathbf0.
\label{eq:fim-increment}
\end{equation}
If the original FIM is positive definite, the CRB cannot increase in
positive-semidefinite order.
\end{proposition}
\begin{IEEEproof}
The added outer product is positive semidefinite, and inversion reverses
positive-definite order. Appendix~\ref{app:information} gives the
corresponding trace reduction. The variance-gradient contribution of an
additional independent path in \eqref{eq:full-gaussian-fim} is also
positive semidefinite.
\end{IEEEproof}

For positive finite path variances, \eqref{eq:fim} further implies
\begin{equation}
\rank(\mathbf F)=2
\quad\Longleftrightarrow\quad
\operatorname{span}\{\mathbf g_{a\ell}\}_{(a,\ell)\in\mathcal A}=\R^2.
\label{eq:rank-condition}
\end{equation}
A single LoS bearing supplies rank-one information and cannot determine
range. Reflections can supply noncollinear gradients because their mirror
transforms differ, enabling local identifiability even with one AP.
The gain depends jointly on gradient orientation and calibration variance:
a path with a large modeled variance contributes less for a fixed gradient.

\begin{remark}[Scope of the information result]
\label{rem:information-scope}
Full rank is a local condition, not global uniqueness under unknown
associations. Correlated errors require a joint covariance, and
deterministic bias cannot be absorbed into a variance term. Consequently,
the ideal CRB does not certify the biased Huber estimator, path matching,
or robustness to unmodeled reflections. We report both ideal PEB and
actual localization RMSE to distinguish these effects.
\end{remark}

\section{Numerical Results}
\label{sec:numerical}
In this section, we evaluate calibration efficiency, localization under
blocked or corrupted paths, and adaptation to environmental change.
All comparisons use the same observations and solver budgets where applicable.
\subsection{Scenario and Baselines}
We consider the $20\times14$ m nominal room in Fig.~\ref{fig:geometry}. Two
APs are located at $[3,2]^T$ and $[17,2]^T$ m. Each AP has one LoS candidate
and reflections from the left, right, and top walls. The physical room is
bounded by the left, right, and top reflection planes at $x=0.35$ m,
$x=19.25$ m, and $y=14.55$ m, respectively. Calibration probes have known
positions and path labels, with independent $0.25^\circ$ angular noise.
All calibration and localization positions are drawn from the same fixed
rectangle $[1.5,17.5]\times[3,12.5]$ m, contained in every tested room.
Changing physical walls therefore does not change the query population or
place queries outside the optimizer's feasible rectangle
$\Omega=[0.5,19.5]\times[0.5,13.5]$ m.
The ET selects the affine or quadratic correction in \eqref{eq:features}
and its ridge strength using \eqref{eq:press-score}. Selection uses only
the labelled reference batch. For an online history, the selected basis is
fixed before the first user query and shared by every update policy.

Unless stated otherwise, the DoA standard deviation is $1.5^\circ$, LoS
blockage probability is 50\%, and 20 calibration probes update the ET. Every
point averages 500 independently drawn user positions, blockage events, angle
errors, and path permutations, conditional on its calibration batch. Query
realizations are paired across algorithms and reused across sweep points
where applicable. The probe-count sweep uses exact prefixes of the same
position and noise pool, rather than merely resetting a seed. No curve is
smoothed or forced to be monotone. Current DoAs are generated from physical
geometry plus Gaussian front-end error. Array-level validation using
spatial channel models such as QuaDRiGa and 3GPP TR~38.901~\cite{quadriga,threeGPP}
remains outside this study.
Table~\ref{tab:parameters} collects the default parameters.

The compared methods are:
\begin{itemize}
\item \emph{LoS + fallback}: bearing intersection with known surviving LoS
labels when at least two independent bearings are present, and a room-center
decision otherwise. The latter is a fallback policy, not a successful
angle-based position estimate.
\item \emph{static geometry}: all available paths and nominal walls, without
physical calibration.
\item \emph{updated ET}: nominal walls plus sparse-probe correction.
\item \emph{oracle geometry}: all paths and physical wall parameters.
\end{itemize}
All map-based methods use the same Huber loss for assignment and position
refinement, the same grid, three candidate starts with 2-m separation,
ten iterations per start, and a common angular weight. The simulations do
not enable covariance weighting, dummy outlier assignments, or innovation
screening. Thus gains should be attributed to calibrated angle predictions
and continued trusted-probe updates, not to these untested extensions.
The RMSE sweeps use linear axes starting at zero. LoS availability and
all-user positioning success
are reported separately, rather than allowing the room-center fallback
error to set the scale of every map-accuracy plot. Complete four-method
RMSE data and plots remain available in the reproducibility files.

The experiments address three distinct questions: whether physical
calibration removes nominal-map bias, whether reflected paths preserve
positioning service under LoS blockage, and whether new probes correct
a stale map after an environmental change. None of these questions requires
the compact ET to outperform every calibrated representation. We therefore
include strong geometric and nonlinear-map baselines in the main noise and
probe-budget sweeps, rather than use separate, repeated sweeps to make the
same comparison. The robust-loss experiment isolates a solver property,
not an advantage unique to the ET representation.

\begin{table}[t]
\centering
\caption{Default Localization Parameters}
\label{tab:parameters}
\begin{tabular}{lc@{\qquad}lc}
\toprule
Parameter & Value & Parameter & Value\\
\midrule
Room size & $20\times14$ m & APs & 2\\
Paths/AP & 4 & DoA std. & $1.5^\circ$\\
Blockage & 50\% & Probes & 20\\
Grid step & 1 m & Realizations & 500\\
Huber threshold & $3^\circ$ & GN iterations & 10\\
Candidate starts & 3 & Separation & 2 m\\
Field dimension & 3 or 6 & Online forgetting & 0.96\\
\bottomrule
\end{tabular}
\end{table}

\subsection{Compared Representations}
The comparison distinguishes measured fingerprints, model-generated
multipath references~\cite{wielandtAoA}, and learned spatial fields. The
implementations, summarized in Table~\ref{tab:diverse-baselines}, share
this paper's unordered-angle interface. They are not
claimed to reproduce the array-spectrum processing, delay observations, or
tracking systems of other papers.

\begin{table}[t]
\centering
\caption{Inputs and Roles of Calibrated Baselines}
\label{tab:diverse-baselines}
\footnotesize
\setlength{\tabcolsep}{4pt}
\begin{tabular}{>{\raggedright\arraybackslash}p{.22\columnwidth}>{\raggedright\arraybackslash}p{.32\columnwidth}>{\raggedright\arraybackslash}p{.34\columnwidth}}
\toprule
Method & Calibration and prior & Main comparison\\
\midrule
Measured WKNN & Probe angles/positions without floor plan & Direct fingerprint lookup\\
Wall-NLS & Same probes and nominal wall model & Correct geometry parameters\\
GP map & Same probes and nominal wall model & Nonlinear residual field\\
ET-LS & Same selected ET field & Quadratic query loss\\
ET-Cauchy & Same selected ET field & Alternative robust loss\\
ET-Huber & Same selected ET field & Reported query solver\\
\bottomrule
\end{tabular}
\end{table}

\emph{Measured weighted $K$-nearest-neighbor (WKNN) fingerprinting}
stores the observed angle set at each reference location. Its distance
$D_j$ is the average squared wrapped-angle difference after minimum-cost
one-to-one matching of each query set to reference $j$. An absent LoS
measurement need not be imputed. The estimate is
\begin{equation}
\widehat{\mathbf p}_{\rm WKNN}=
\frac{\sum_{j\in\mathcal N_K}w_j\mathbf p_j}
{\sum_{j\in\mathcal N_K}w_j},\quad
w_j=\frac{1}{\max\{D_j,10^{-8}\}},
\label{eq:wknn-baseline}
\end{equation}
where the floor prevents division by zero.
The neighbor count is selected from $\{1,3,5\}$ by leave-one-reference-out
position error, averaging the four LoS-visibility patterns of the two APs.
No test-user position is used. Unlike a dense fingerprint database, this
baseline is given the same sparse physical probe budget.

\emph{Wall nonlinear least squares (Wall-NLS)} calibrates the left, right,
and top reflection-plane coordinates $\mathbf w$ directly:
\begin{equation}
\widehat{\mathbf w}=\arg\min_{\mathbf w}
\sum_{a,\ell,i}\left[
\operatorname{wrap}\{\widetilde\theta_{a\ell,i}
-\theta^{\rm geo}_{a\ell}(\mathbf p_i;\mathbf w)\}\right]^2.
\label{eq:wall-calibration}
\end{equation}
The room must contain the known reference positions. A three-variable
Nelder--Mead search starts from the nominal floor plan, without access to
physical-wall truth. The fitted geometry then applies the same Huber association
and refinement as the ET. This is a particularly strong baseline here:
its parameterization matches the simulator's planar reflection model.

\emph{Gaussian-process (GP) residual mapping} replaces the finite-order
correction with a squared-exponential spatial kernel. For each
labelled path, its posterior-mean correction is~\cite{rasmussenGP}
\begin{equation}
\widehat d(\mathbf p)=\overline d+
\mathbf k(\mathbf p)^T(\mathbf K+\sigma_c^2\mathbf I)^{-1}
(\mathbf d-\overline d\mathbf 1).
\label{eq:gp-baseline}
\end{equation}
The empirical mean and hyperparameters use calibration data only. Kernel
lengths are selected from $\{2,4,8,16,32\}$ m and signal variances from
$\{0.25,1,4,16\}\max\{\operatorname{var}(\mathbf d),\sigma_c^2\}$ by
training marginal likelihood. Query positions are estimated by the same
Huber solver. Predictive variance is not used to introduce a second change
to the comparison.

\emph{ET-LS and ET-Cauchy} keep the selected calibration unchanged and replace
only the angular query loss. The alternatives are
\begin{equation}
\rho_Q(r)=\frac{r^2}{2},\qquad
\rho_C(r)=\frac{c^2}{2}\log\!\left(1+\frac{r^2}{c^2}\right),
\end{equation}
respectively, with $c=3^\circ$.
Each of the assignment, iteratively reweighted refinement, and
backtracking stages uses the selected loss consistently. These are ablations of the query solver,
not independent new map architectures. The nominal Huber map and
physical-geometry Huber reference remain in the comparison.
The ``Selected ET'' legend denotes ET-Huber with the field chosen by
\eqref{eq:press-score}, not a per-query selection among competing methods.

The noise and probe-budget sweeps below use the same fixed 80-probe pool
and 500 held-out queries per point. Probe budgets are exact nested
prefixes, including measurement noise. Every method sees identical query
positions, angles, and blockage events, while the map-based solvers use
the same initialization grid, multistart rule, and iteration budget. WKNN
selection, GP fitting, and ET model selection use calibration data only.
The other controlled sweeps retain
their original paired samples. Numerical differences across figures should
not be interpreted as comparisons on one identical calibration batch.

\subsection{Effect of DoA Noise}
\begin{figure}[t]
\centering
\includegraphics[width=.97\columnwidth]{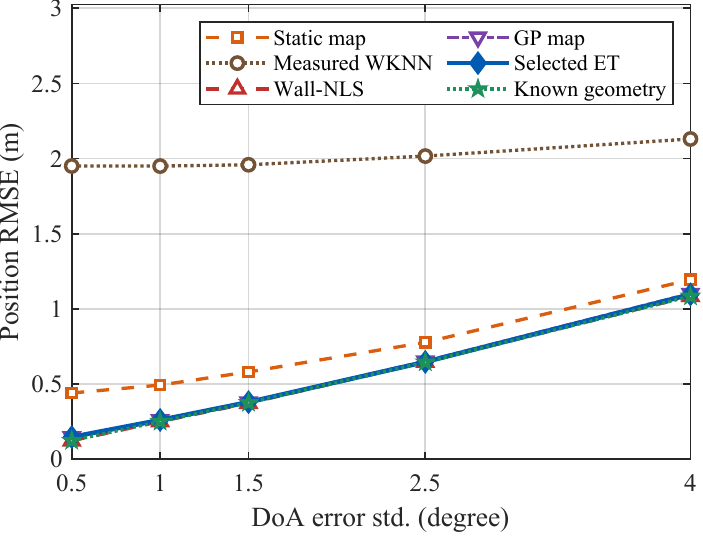}
\caption{DoA-noise sensitivity with calibrated baselines.}
\label{fig:doa-noise}
\end{figure}

Fig.~\ref{fig:doa-noise} distinguishes two regimes. At low DoA noise,
the nominal map is limited by geometry mismatch, whereas physical
calibration substantially reduces this bias. At higher noise, the
calibrated-map curves approach one another because angular measurement
error increasingly controls the position estimate. Improving the map and
the DoA front end therefore address different error sources.
The ET lies close to the calibrated geometric and GP alternatives, but
does not dominate them: at $1.5^\circ$, Wall-NLS, GP mapping, and ET-Huber
give \LocDiverseWall, \LocDiverseGp, and \LocDiverseHuber\ m,
respectively. In particular, Wall-NLS is well matched to the three-plane
simulator. This figure supports the use of a calibrated directional state,
not a unique accuracy advantage of its polynomial parameterization. At high
noise, choosing a richer field or a lower query objective need not improve
position RMSE. The selection criterion does not optimize unknown test locations.

The common-batch association diagnostic gives correct-association
probabilities above 94\% for the nominal, calibrated, and physical-geometry
maps even at $4^\circ$ DoA error, with little separation between them.
Thus, the calibration gain mainly reflects more accurate predicted angles,
not a substantially better combinatorial matcher.

\subsection{LoS Blockage}
\begin{figure}[t]
\centering
\includegraphics[width=.97\columnwidth]{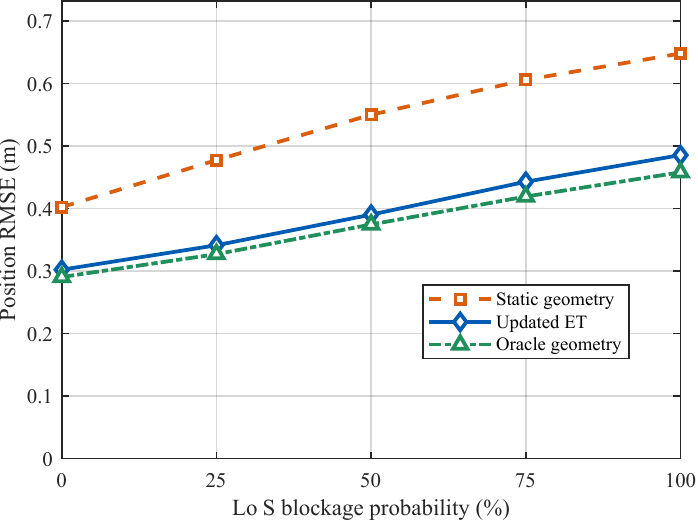}
\caption{Map-based position RMSE versus LoS blockage probability.}
\label{fig:blockage}
\end{figure}

In Fig.~\ref{fig:blockage}, static geometry and the calibrated ET obtain
\LocStaticBlockHalf\ and \LocEtBlockHalf\ m RMSE at 50\% LoS blockage.
With complete LoS blockage, their errors are \LocStaticBlockFull\ and
\LocEtBlockFull\ m. The relevant gain is the ability to use reflected
directions after LoS constraints disappear, together with the correction of
nominal-wall bias.

LoS performance must be interpreted separately. With two independently blocked
LoS links and per-link blockage probability $b$, bearing availability is
$(1-b)^2$, not one. At $b=0.5$, empirical availability is
\LocLosAvailability\%, and conditional RMSE on those available cases is
\LocLosConditional\ m. The all-query value \LocLosBlockHalf\ m includes
room-center fallback on unavailable cases. Its large value cannot be
presented as the error of a functioning LoS estimator. Availability and
conditional RMSE are exported separately with the numerical results.

\begin{figure}[t]
\centering
\includegraphics[width=.97\columnwidth]{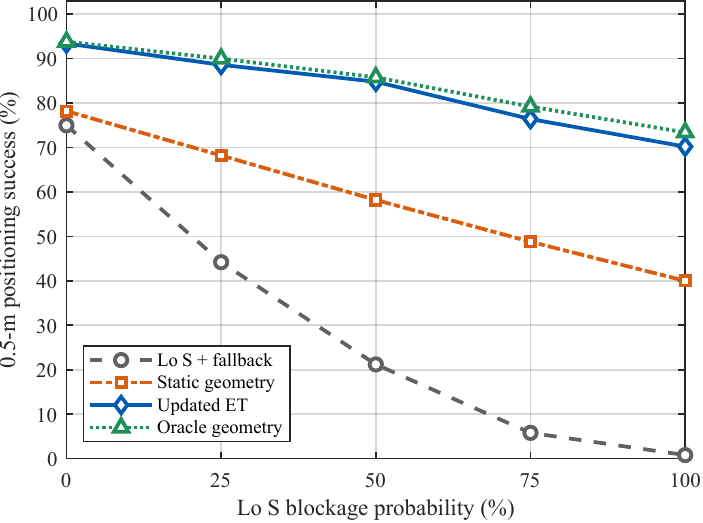}
\caption{All-user probability of position error within 0.5 m.}
\label{fig:service-success}
\end{figure}

Fig.~\ref{fig:service-success} asks a different, service-oriented question:
what fraction of users can be positioned within a specified tolerance?
For the same 500 queries used in Fig.~\ref{fig:blockage}, we report
\begin{equation}
\widehat P_{0.5}=\frac{1}{N}\sum_{n=1}^{N}
\mathbf 1\{\|\widehat{\mathbf p}_n-\mathbf p_n\|_2\leq0.5\ {\rm m}\}.
\label{eq:service-success}
\end{equation}
The 0.5-m tolerance is an illustrative indoor requirement, not a claimed
standard or an optimized threshold. All users remain in the denominator,
including unavailable LoS cases. Even a coincidentally accurate fallback
counts as a success. At 50\% blockage, LoS with fallback, the nominal map,
and the calibrated ET give \LocServiceLosHalf\%,
\LocServiceStaticHalf\%, and \LocServiceEtHalf\%, respectively.
With all LoS paths blocked, the latter two give
\LocServiceStaticFull\% and \LocServiceEtFull\%. This comparison separates
the value of retaining reflected-path constraints from the additional
value of calibrating their directions. The ideal-geometry reference remains
visible, so that the remaining calibration gap is not hidden.

\subsection{Calibration Budget and Model Mismatch}
\begin{figure}[t]
\centering
\includegraphics[width=.97\columnwidth]{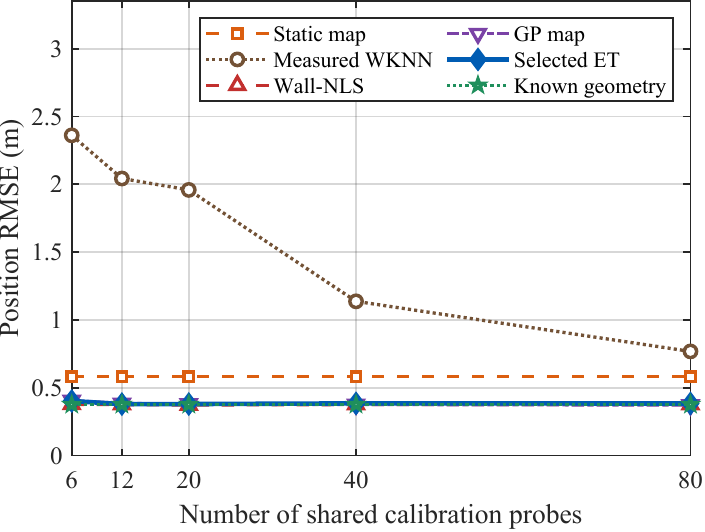}
\caption{Position RMSE versus shared calibration budget.}
\label{fig:calibration}
\end{figure}

Fig.~\ref{fig:calibration} asks how effectively each representation uses
the same physical probes. The calibrated model-based methods work with
relatively few probes, whereas sparse fingerprint lookup improves more
gradually as reference locations cover the room. At 80 probes, WKNN gives
\LocKnnDense\ m, compared with \LocGpDense\ m for the GP and
\LocHuberDense\ m for the selected ET. This difference reflects both
spatial interpolation and prior information: WKNN is not supplied with
the floor plan. It is not evidence of superiority over dense fingerprinting.

The more important comparison is within the calibrated maps. Wall-NLS
employs a correct physical parameterization, GP mapping learns a flexible
residual field, and the ET retains only three or six coefficients per path.
In this nested calibration pool, LOOCV selects the affine field with
$\lambda_\beta=0.1$ at six probes, and the quadratic field with
$\lambda_\beta=0.01$ at the other budgets. This change is decided from
held-out angular innovations, not localization errors. Their close errors
show that a compact correction captures much of the useful map information
in this scene. A correctly specified wall model can still do better.
Additional probes reduce estimation uncertainty but do not remove query
noise or all approximation error. Small nonmonotonic changes are retained. Noisy finite
calibration does not guarantee monotonically decreasing test RMSE.

\begin{figure}[t]
\centering
\includegraphics[width=.97\columnwidth]{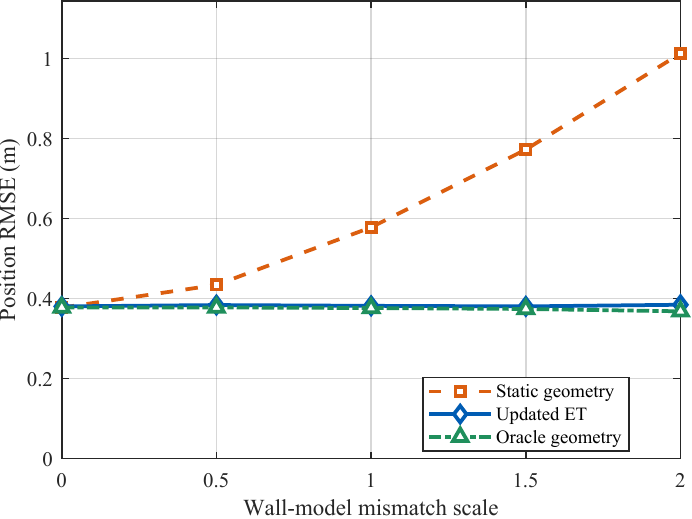}
\caption{Position RMSE versus physical-to-nominal wall mismatch scale.}
\label{fig:wall-mismatch}
\end{figure}

Fig.~\ref{fig:wall-mismatch} scales physical wall displacement from zero
to twice its default value. At the largest mismatch, static geometry gives
\LocWallStaticEnd\ m and the calibrated ET gives \LocWallEtEnd\ m.
At zero mismatch, calibration is unnecessary and can introduce an
estimation cost by fitting probe noise.
Thus the useful regime is an informative but imperfect floor plan, not an
already exact map. Increasing angular noise can likewise conceal the value
of improved geometry, as Fig.~\ref{fig:doa-noise} demonstrates.
Each scene receives its own trusted calibration probes, so this is a
recalibration test rather than zero-shot generalization to unseen walls.
The remaining gap to the perfect-map reference reflects both finite probe
noise and the limited spatial flexibility of the correction field.

\subsection{Robustness to Gross Angle Errors}

\begin{figure}[t]
\centering
\includegraphics[width=.97\columnwidth]{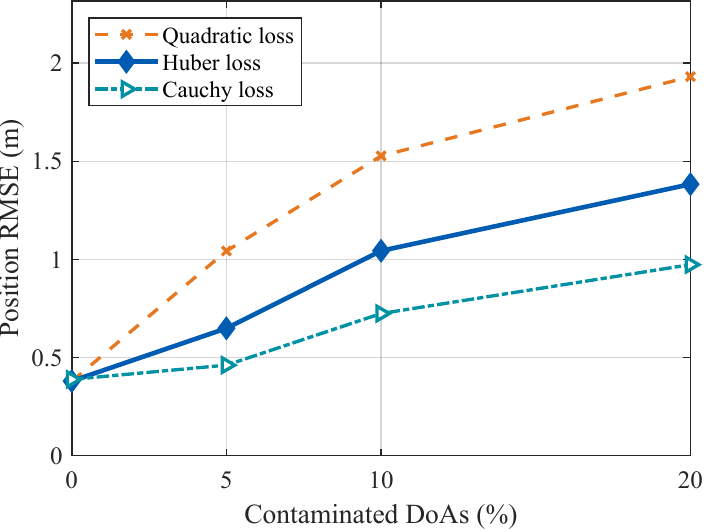}
\caption{Loss-function comparison on a common calibrated map.}
\label{fig:robust-losses}
\end{figure}

Fig.~\ref{fig:robust-losses} isolates robustness from representation. Each
observed query angle independently receives an additional zero-mean
$20^\circ$-standard-deviation Gaussian perturbation with probability
$p_o\in\{0,0.05,0.10,0.20\}$, on top of the common $1.5^\circ$ noise.
The physical probes remain clean, the path count is unchanged, and all
methods receive exactly the same corrupted observations. This controlled
gross-angle-error test is not a simulation of false detections or unresolved
coherent sources. At $p_o=0.20$, quadratic, Huber, and Cauchy query losses
give \LocContamLs, \LocContamHuber, and \LocContamCauchy\ m RMSE.
The three losses are close without gross errors, so Gaussian-noise curves
alone provide little evidence for a robust-loss advantage. In this
contamination test, Huber improves on quadratic loss, while the fixed-scale
Cauchy alternative yields the lowest RMSE at the tested nonzero contamination
levels. Its redescending influence suppresses very large angular residuals
more strongly. Huber is retained as the fixed default throughout the other
experiments, not switched to Cauchy after inspecting contaminated test cases.
A convex residual penalty is convenient, but it does not make the joint
position-and-association problem convex or make Huber uniformly optimal.
This is a conditional robustness comparison, not an optimality claim across
noise models or tuning choices. The residual
error remains substantial because corrupted angles can still mislead the
initialization or path assignment. Robust refinement does not solve arbitrary
outlier association.

\subsection{Number of APs and Paths}
\begin{figure}[t]
\centering
\includegraphics[width=.97\columnwidth]{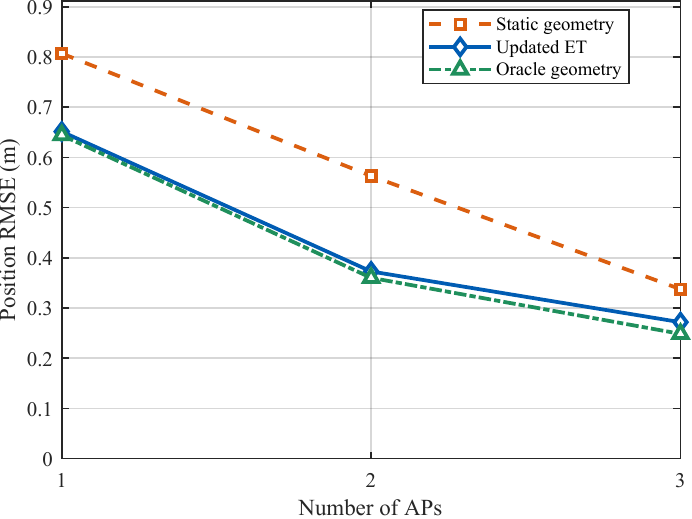}
\caption{Position RMSE versus number of available APs.}
\label{fig:ap-count}
\end{figure}

Fig.~\ref{fig:ap-count} tests whether reflected paths compensate for fewer
APs. With one AP, a single LoS bearing does not identify range. Its all-query
baseline comes entirely from the fallback policy, and is omitted from this
map-only RMSE plot. The calibrated multipath map
instead gives \LocOneApEt\ m because reflected-path gradients can supply
independent spatial constraints. Three APs reduce its error to
\LocThreeApEt\ m. This result depends on the known reflector model and
resolvable reflected paths. It does not imply generic single-AP
identifiability in an arbitrary environment.

\begin{figure}[t]
\centering
\includegraphics[width=.97\columnwidth]{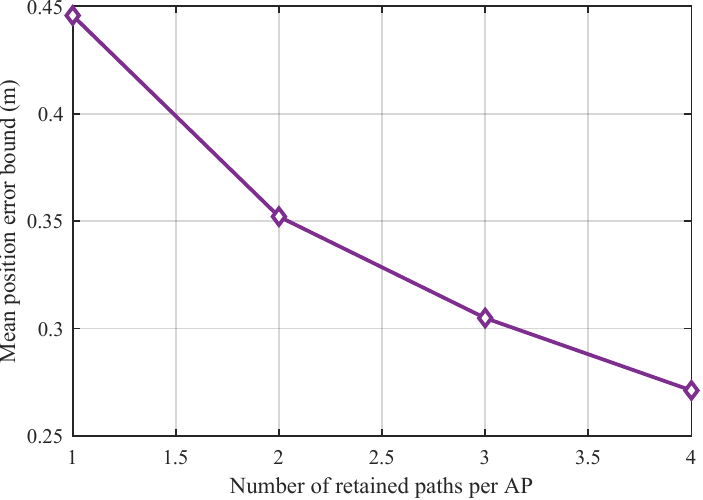}
\caption{Mean ideal-geometry PEB versus number of retained paths per AP.}
\label{fig:fim-paths}
\end{figure}

The information calculation in Fig.~\ref{fig:fim-paths} uses correct
associations and physical geometry. Mean PEB decreases from \LocPebOne\ m
with one path per AP to \LocPebFour\ m with all four. The monotonic trend
matches Proposition~\ref{prop:path-information}. These values use all retained paths without random
LoS blockage and therefore must not be read as matched lower bounds on the
blocked, biased, unknown-association experiments. They isolate the effect
of additional ideal geometric constraints.

\subsection{Persistent Recovery after Environmental Change}
\begin{figure}[t]
\centering
\includegraphics[width=.97\columnwidth]{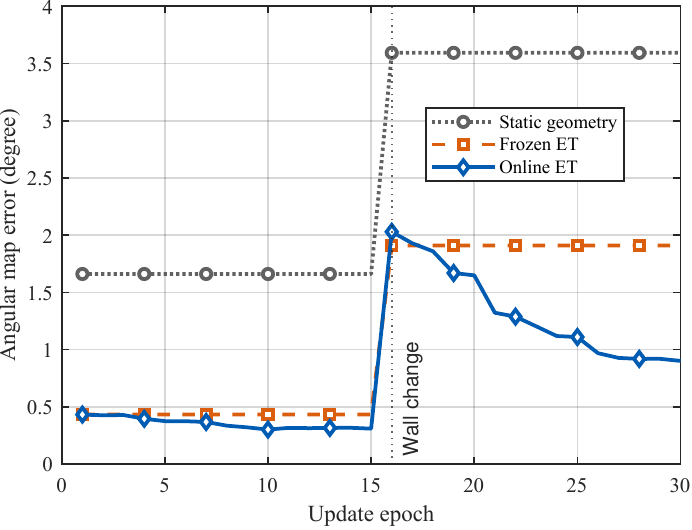}
\caption{Angular correction-map error during persistent operation. The
physical walls change at epoch 16, and two trusted probes update the online ET
after every epoch.}
\label{fig:online-map}
\end{figure}

The preceding experiments recalibrate the ET for each fixed environment.
Fig.~\ref{fig:online-map} instead retains one state across 30 epochs. Online
and frozen ETs start from the same 12-probe batch calibration. At epoch 16,
the left, right, and top reflection planes change to $x=0.72$ m,
$x=18.55$ m, and $y=15.15$ m. Localization and map evaluation occur before
that epoch's two new trusted probes are assimilated, so neither method can
anticipate the change. RLS uses $\lambda_f=0.96$ and the initial batch
normal matrix, with no innovation gate or change-triggered reset. Its
initially selected basis remains fixed throughout the 30 epochs.
The online map error changes from $\LocMapJump^\circ$ at the change to
$\LocMapEnd^\circ$ at the final epoch, compared with
$\LocFrozenMapEnd^\circ$ for the frozen map. Their different errors at the
change reflect their different pre-change coefficients, not Monte Carlo
variation in the deterministic map metric.

For the same calibration sequence, position RMSE decreases from
\LocOnlineJump\ m at epoch 16 to \LocOnlineEnd\ m at epoch 30,
versus \LocFrozenEnd\ m for the frozen map and \LocStaticOnlineEnd\ m
for nominal geometry. Each epoch uses 500 queries paired across methods
and epochs through common random numbers. This is one evolving calibration
history, not 500 independent map histories or a tracked user trajectory.
The next experiment tests localization over independent update histories.

\begin{remark}[Querying without self-calibration]
Current user DoAs are always used for positioning, but only trusted probes are
used to modify the map in this experiment. A high-confidence user estimate
would require the additional tests discussed in Section~\ref{sec:trusted-update}. Keeping the query
path and update path separate is essential: a doubtful position may benefit
from the ET without being allowed to corrupt future positions.
\end{remark}

\subsection{Comparison of Online Update Policies}
An online-versus-frozen comparison alone cannot establish an advantage
over other ways of using new probes. We therefore compare three updates
of the same correction field: all-history least squares (LS), LS on the
most recent 12 probes, and RLS with $\lambda_f=0.96$. All start from the
same 12-probe calibration and receive the same two new labelled probes
after each epoch. The frozen map and physical-geometry reference are
retained as controls. The query grid, association, robust refinement,
and DoA observations are identical for every method.

For all-history and sliding-window LS, let $\mathcal I_t$ contain all
available probes or the last 12, respectively. Each solves
\begin{equation}
\widehat{\boldsymbol\beta}_{a\ell,t}=
\arg\min_{\boldsymbol\beta}
\sum_{i\in\mathcal I_t}
(d_{a\ell,i}-\boldsymbol\phi_i^T\boldsymbol\beta)^2
+\widehat\lambda\|\boldsymbol\beta\|_2^2.
\label{eq:update-baselines}
\end{equation}
Here, $\widehat\lambda$ and the basis are selected from each history's
initial batch and then shared and frozen across all update policies.
The RLS recursion instead progressively discounts older information.
Its initial ridge contribution decays with the initial normal matrix.
There is no oracle change detector, reset at epoch 16, or access to future
probes. The window length and forgetting factor are fixed rather than
selected to favor one method on the test queries. All three adaptive
methods implement the persistent-state framework. The comparison tests
memory policies, not an ET label against otherwise identical LS.

\begin{figure}[t]
\centering
\includegraphics[width=.97\columnwidth]{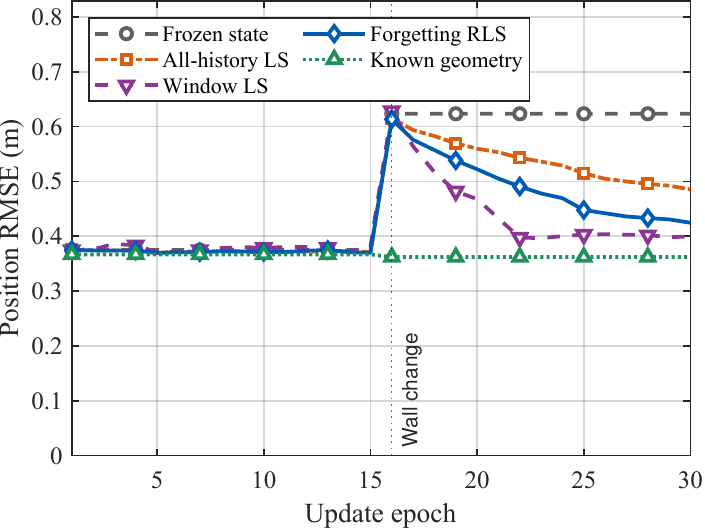}
\caption{Update-policy comparison with identical probe streams.}
\label{fig:update-benchmark}
\end{figure}

Fig.~\ref{fig:update-benchmark} averages squared error over 10 independent
calibration/query histories with 50 queries each at every epoch, then
takes the square root. All histories use the same prescribed wall-change
schedule. They are not 10 different environments. The first 15 epochs
provide an unchanged-scene control, while epochs 16--30 test the use of
stale measurements. The close pre-change curves are a useful control:
updating an already accurate state should not by itself create a large
gain. Once the walls change, the frozen state retains an outdated angular
field. All-history LS incorporates new evidence, but old observations
continue to pull its fit toward the previous environment. RLS reduces
this influence gradually, whereas the 12-probe window removes old
evidence sooner. The post-change RMSE is \LocFairWindowPost\ m for the
window and \LocFairRlsPost\ m for RLS.

\begin{remark}[Representation versus memory policy]
The frozen-versus-adaptive comparison supports maintaining the stored
path-angle state. Differences among adaptive curves concern memory policy,
not the ET representation itself. The configured RLS is not the optimum among the tested policies
for this abrupt change. A short window discards stale evidence faster but
also loses potentially useful probes, whereas RLS discounts them gradually.
The present result does not determine the best policy for gradual motion,
noisier probes, or other scenes.
\end{remark}

\section{Conclusion}
\label{sec:conclusion}
This paper formulated multipath DoA localization as a query to a persistent
electromagnetic state. Nominal virtual-anchor geometry provided physical
structure, calibration-only model selection set the correction-field order
and ridge strength, and continued trusted updates maintained that correction.
Assignment-aware multistart search retained the original initialization
while examining other association basins. Joint
association and robust position refinement remained computationally simple
because the path set was small and the continuous variable was two
dimensional. The analysis and simulations support a precise conclusion: correctly
associated reflected paths can add positioning information and preserve rank
under LoS blockage, but their benefit depends on physical calibration.  

\appendices
\section{Virtual-Path Jacobian}
\label{app:jacobian}
Let
\begin{equation}
\mathbf d_{a\ell}(\mathbf p)=
\mathbf H_{a\ell}\mathbf p+\mathbf h_{a\ell}-\mathbf q_a
=[d_x,d_y]^T.
\end{equation}
For angle $\theta=\operatorname{atan2}(d_y,d_x)$,
\begin{equation}
\frac{\partial\theta}{\partial\mathbf d^T}
=\frac{1}{d_x^2+d_y^2}[-d_y,\ d_x].
\end{equation}
Therefore,
\begin{equation}
\nabla_{\mathbf p}\theta_{a\ell}^{\rm geo}
=\mathbf H_{a\ell}^T
\frac{[-d_y,\ d_x]^T}{d_x^2+d_y^2},
\end{equation}
and the ET correction adds
$\nabla_{\mathbf p}(\boldsymbol\phi^T\boldsymbol\beta_{a\ell})$.
For the quadratic basis, suppressing path indices and numbering coefficients
from one, the additional derivatives are
\begin{equation}
\begin{split}
\frac{\partial(\boldsymbol\phi_2^T\boldsymbol\beta)}{\partial p_x}
&=\frac{\beta_2+2u\beta_4+v\beta_5}{W},\\
\frac{\partial(\boldsymbol\phi_2^T\boldsymbol\beta)}{\partial p_y}
&=\frac{\beta_3+u\beta_5+2v\beta_6}{H}.
\end{split}
\end{equation}
The affine case follows by setting $\beta_4=\beta_5=\beta_6=0$.
The geometric derivatives above are in radians, while correction derivatives
inherit the units of their coefficients. Both terms must use the same unit
before addition. The implementation uses degrees for angular derivatives
and noise standard deviations, preserving the same information after
consistent unit conversion.
The implementation employs symmetric numerical differentiation so that the same
solver also supports the quadratic correction and the GP baseline.

\section{Map-Uncertainty Propagation}
\label{app:uncertainty}
Let
$\widehat{\boldsymbol\beta}_{a\ell}
=\boldsymbol\beta_{a\ell}+\widetilde{\boldsymbol\beta}_{a\ell}$,
with error mean $\mathbf b_{a\ell}$ and covariance
$\mathbf P_{\beta,a\ell}$. Linear correction error at position $\mathbf p$ is
\begin{equation}
e_{\rm map}(\mathbf p)=
\boldsymbol\phi^T(\mathbf p)
\widetilde{\boldsymbol\beta}_{a\ell}.
\end{equation}
Hence
\begin{align}
\E[e_{\rm map}]&=\boldsymbol\phi^T(\mathbf p)\mathbf b_{a\ell},\\
\operatorname{var}(e_{\rm map})
&=\boldsymbol\phi^T(\mathbf p)\mathbf P_{\beta,a\ell}
\boldsymbol\phi(\mathbf p).
\end{align}
The mean-square error additionally contains
$(\boldsymbol\phi^T\mathbf b_{a\ell})^2$.
If front-end angle error, map error, and residual model error are independent,
their variances add, giving \eqref{eq:effective-variance}. A nonzero mean
remains a separate bias and is not removed by variance weighting.

\section{Information Monotonicity}
\label{app:information}
Let $\mathbf F_0\succ\mathbf0$ be information from an existing correctly modeled path
set and $\Delta\mathbf F\succeq\mathbf0$ be \eqref{eq:fim-increment}. Then
\begin{equation}
\mathbf F_0+\Delta\mathbf F\succeq\mathbf F_0.
\end{equation}
Matrix inversion reverses the positive-definite order:
\begin{equation}
(\mathbf F_0+\Delta\mathbf F)^{-1}
\preceq\mathbf F_0^{-1}.
\end{equation}
Taking trace proves that the squared PEB cannot increase. For one
finite-variance nonzero gradient, the reduction is strictly positive. The
matrix inversion lemma gives
\begin{equation}
\tr(\mathbf F_0^{-1})-
\tr[(\mathbf F_0+\mathbf g\mathbf g^T/v)^{-1}]
=\frac{\mathbf g^T\mathbf F_0^{-2}\mathbf g}
{v+\mathbf g^T\mathbf F_0^{-1}\mathbf g}.
\end{equation}
Equality requires $\mathbf g=\mathbf0$ for finite $v$, or is approached as
$v\to\infty$. The rank condition for an initially singular matrix must be
checked separately.

\section{Calibration Rank, Error, and Recursive Form}
\label{app:calibration}
For the unregularized batch model, suppose two coefficients
$\boldsymbol\beta_1$ and $\boldsymbol\beta_2$ explain the same noise-free probe
innovations. Then
\begin{equation}
\boldsymbol\Phi_K(\boldsymbol\beta_1-\boldsymbol\beta_2)=\mathbf0.
\end{equation}
Uniqueness requires a trivial null space of $\boldsymbol\Phi_K$,
equivalently \eqref{eq:calibration-rank}. Under full column rank,
ordinary LS has error
\begin{equation}
\widetilde{\boldsymbol\beta}
=(\boldsymbol\Phi_K^T\boldsymbol\Phi_K)^{-1}
\boldsymbol\Phi_K^T\boldsymbol\epsilon
\end{equation}
and covariance
$\sigma_c^2(\boldsymbol\Phi_K^T\boldsymbol\Phi_K)^{-1}$. Taking the trace and
using
$\tr[(\boldsymbol\Phi_K^T\boldsymbol\Phi_K)^{-1}]
\leq D/\mu_K$
proves \eqref{eq:calibration-bound}.

For recursive operation, define the exponentially weighted normal matrix
\begin{equation}
\mathbf G_k=
\sum_{i=1}^{k}\lambda_f^{k-i}
\boldsymbol\phi_i\boldsymbol\phi_i^T+
\lambda_f^k\mathbf G_0.
\end{equation}
Applying the matrix inversion lemma to
$\mathbf A_k=\mathbf G_k^{-1}$ yields the gain and inverse-normal recursion in
\eqref{eq:rls}. Persistent excitation prevents $\mathbf G_k$ from losing rank.
The right-hand side of the normal equation is
$\mathbf h_k=\lambda_f^k\mathbf G_0\widehat{\boldsymbol\beta}_0+
\sum_{i=1}^k\lambda_f^{k-i}\boldsymbol\phi_i d_i$.
Substituting $d_i=\boldsymbol\phi_i^T\boldsymbol\beta^\star+\epsilon_i$
and taking conditional expectation results in \eqref{eq:rls-mean-error}.
For independent noise, the coefficient sampling covariance instead requires
\begin{equation}
\mathbf P_{\beta,k}=\mathbf A_k\left[
\lambda_f^{2k}\mathbf G_0\mathbf P_{\beta,0}\mathbf G_0+
\sigma_c^2\sum_{i=1}^k\lambda_f^{2(k-i)}
\boldsymbol\phi_i\boldsymbol\phi_i^T\right]\mathbf A_k^T.
\end{equation}
This covariance assumes independence between the initial fit and subsequent
probe noise and excludes deterministic model bias. Thus the RLS inverse
normal matrix alone is not a calibrated uncertainty estimate.

\end{document}